\documentclass[a4paper,leqno]{amsart}

\date{October 29, 2016}

\usepackage{amsmath}
\usepackage{amsthm}
\usepackage{amssymb}
\usepackage{amsfonts}
\usepackage[applemac]{inputenc}
\usepackage{mathrsfs}
\usepackage{pdfsync}
\usepackage{color}

\def\C{{\mathbb C}}
\def\R{{\mathbb R}}
\def\N{{\mathbb N}}

\def\le{\leqslant}
\def\ge{\geqslant}

\newcommand{\im}{\mathrm{Im}}

\theoremstyle{plain}
\newtheorem{theorem}{Theorem}[section]
\newtheorem{lemma}[theorem]{Lemma}

\newtheorem{proposition}[theorem]{Proposition}

\theoremstyle{definition}

\newtheorem{remark}[theorem]{Remark}
\newtheorem*{remark*}{Remark}

\def\RR{{\mathbb R}}

\def\eps{\varepsilon}

\def\pa{\partial}

\def\calH{{\mathcal H}}

\def\bx{{\mathbf x}}

\numberwithin{equation}{section}

\begin{document}

\title[Strong magnetic confinement in NLS]
{Averaging of nonlinear Schr\"odinger equations with strong magnetic confinement}

\author[R. Frank]{Rupert L. Frank}
\address[R. Frank]{Mathematisches Institut, Ludwig-Maximilans Universit\"at M\"unchen, Theresienstr. 39, 80333 M\"unchen, Germany, and Department of Mathematics, California Institute of Technology, 
Pasadena, CA 91125, USA}
\email{rlfrank@caltech.edu}

\author[F. M\'ehats]{Florian M\'ehats}
\address[F. M\'ehats]{IRMAR, Universit\'e de Rennes 1 and INRIA, IPSO Project, Campus de Beaulieu, 
35042 Rennes Cedex, France}
\email{florian.mehats@univ-rennes1.fr}

\author[C. Sparber]{Christof Sparber}
\address[C.~Sparber]
{Department of Mathematics, Statistics, and Computer Science, M/C 249, University of Illinois at Chicago, 851 S. Morgan Street, Chicago, IL 60607, USA}
\email{sparber@math.uic.edu}

\begin{abstract}
We consider the dynamics of nonlinear Schr\"odinger equations with strong constant magnetic fields. 
In an asymptotic scaling limit the system exhibits a purely magnetic confinement, based on the 
spectral properties of the Landau Hamiltonian. Using an averaging technique we derive 
an associated effective description via an averaged model of nonlinear Schr\"odinger type. In a special case this 
also yields a derivation of the LLL equation.
\end{abstract}


\thanks{This publication is based on work supported by the U.S. National Science Foundation through grants DMS-1363432 (R.L.F.) and DMS-1348092 (C.S.) and by the ANR project Moonrise ANR-14-CE23-0007-01 (F.M.).}

\maketitle

\section{Introduction}
\label{sec:intro}

In this work, we study the asymptotic scaling limit, as $\eps \to 0_+$ of nonlinear Schr\"odinger equations (NLS) with strong magnetic fields. 
Such equations arise, e.g., in the macroscopic description of fermion pairs in the state of Bose-Einstein condensation, cf. \cite{HS}.
To be more precise, we consider the following 
NLS-type model (in dimensionless units) in three spatial dimensions:
\begin{equation}
\label{NLSinit}
i\pa_t \psi=\frac{1}{2}\big( - i \nabla_\bx  + A^\eps(\bx) \big)^2 \psi + V(z) \psi + \beta^\eps |\psi|^{2\sigma}\psi,
\end{equation}
where $(t,\bx)\in \R\times \R^3$, $\sigma \in \N$, and $\beta^\eps \in \R$, some nonlinear coupling constant (to be made precise later on). We will denote the spatial degrees of freedom by 
$\bx = (x_1, x_2, z)\in \R^3$ and also write $x=(x_1,x_2)$, for simplicity. The real-valued potential $V$ is assumed to be smooth and sub-quadratic, i.e., 
for $\alpha\ge 2$: 
\begin{equation} \label{Vhyp} |\partial^\alpha V (z)| \le C_\alpha, \quad \text{for all $z\in \R$.}
\end{equation}
A possible example would be a harmonic oscillator potential in $z$-direction. The vector potential $A^\eps$ is assumed to be given by
\[
A^\eps(\bx)= \frac{1}{2\eps^2} (-x_2, x_1, 0), 
\]
where $0< \eps \ll1$ is a small (adiabatic) parameter. This implies that $A^\eps$ is divergence free $\nabla_\bx \cdot A^\eps =0$, and hence
\[
\big( - i \nabla_\bx + A^\eps(\bx) \big)^2 = - \Delta_{\bf x} + \frac{1}{4 \eps^4} |x|^2 {-}  \frac{i}{\eps^2} (x_1 \partial_{x_2} - x_2 \partial_{x_1}).
\]
The corresponding magnetic field is given by 
\[
B^\eps = \nabla \times A^\eps = \frac{1}{\eps^2}(0,0,1)\in \R^3, 
\]
i.e., a constant magnetic field in the $z$-direction with field strength $|B^\eps| =\frac{1}{\eps^2}\gg 1$. 

We want to analyze the strong magnetic confinement limit as $\eps \to 0_+$, assuming 
that the initial data for \eqref{NLSinit} is of the form \[\psi(0,\bx)={\eps}^{-1} \psi_0\left(\frac{x}{\eps},z\right), \quad \text{with $\| \psi_0 \|_{L^2}=1$.}\] In other words, we assume that the initial wave function is 
already confined at the scale epsilon in the $x$-directions (an assumption which is consistent with the asymptotic limiting regime considered). To this end, we rescale 
\[
x'=\frac{x}{\eps}, \quad z'=z, \quad \psi^\eps(t,x',z')=\eps \psi \left(t,\eps x', z' \right),
\]
ensuring that $\| \psi^\eps\|_{L^2} = \|\psi \|_{L^2}=1$. Moreover, we shall assume that the nonlinear coupling constant is of the form 
$\beta^\eps=\lambda \eps^{2\sigma}\ll 1$, where $\lambda\in \RR$ is fixed. We are thus in a {\it weak interaction regime}.

In the rescaled variables, equation \eqref{NLSinit}
becomes (dropping the primes in the variables for simplicity)
\begin{equation}\label{NLS}
i\pa_t \psi^\eps=\frac{1}{\eps^2} \mathcal H \psi^\eps -\frac{1}{2} \partial^2_z \psi^\eps + V(z) \psi^\eps + \lambda |\psi^\eps|^{2\sigma}\psi^\eps,\quad \psi^\eps|_{t=0}=\psi_0(x,z),
\end{equation}
where, using $x^\perp = (-x_2, x_1)$, we denote 
\[
\mathcal H = \frac{1}{2} \left( - i \nabla_x  + \frac{1}{2} x^\perp \right)^2= -\frac{1}{2} \Delta_x + \frac{1}{8}|x|^2 {-}   \frac{i}{2} x^\perp \cdot \nabla_x, 
\]
i.e., the classical Landau Hamiltonian in symmetric gauge, cf. \cite{LL}. Note that $\mathcal H$ commutes with the rest of the linear Hamiltonian, i.e., 
$$[\mathcal H, \partial_z^2]=[\calH, V]=0.$$
We comment on the case where $V$ depends not only on $z$ briefly in Remark \ref{extension} below.

Let us recall some well-known facts about the spectral properties of $\mathcal H$. 
One finds that $\mathcal H$ is essentially self-adjoint on $C_0^\infty (\R^2)\subset L^2(\R^2)$ with pure point spectrum given by
\begin{equation}\label{spec}
\text{spec} \, \mathcal H = \Big\{n+\frac{1}{2}\Big\}, \quad n\in \N_0. 
\end{equation}
These are the same eigenvalues as for a one-dimensional harmonic oscillator. Each $n\in \N_0$ thereby corresponds to a distinct Landau level. 
In contrast to the harmonic oscillator, however, the corresponding eigenspaces $P_n L^2(\R^2)$, are {\it infinitely degenerate}. Here, and in the following, 
we denote by $P_n=P_n^2$ the spectral projection in $L^2(\R^2)$ onto the $n$-th eigenspace of $\mathcal H$.

\smallskip

In view of \eqref{NLS}, it is clear that $\mathcal H$ induces high frequency oscillations (in time) $\propto \mathcal O(\eps^{-2})$ within the solution $\psi^\eps$. 
By filtering these oscillations, we consequently expect the following limit in a strong norm,
\[
\phi^\eps(t,\bx) := e^{i t \mathcal H/\eps^2} \psi^\eps(t,\bx)\stackrel{\eps\rightarrow 0_+ }{\longrightarrow}\phi(t,\bx).
\]
In order to describe the behavior of the limit $\phi$, a natural functional framework is given by the space,
$$\Sigma^2:=\left\{u\in H^2(\RR^3)\,:\,|\bx|^2u\in L^2(\RR^3)\right\},$$
which is equipped with the norm
\begin{equation}\label{norm}
\|u\|_{\Sigma^2}:= \left( \|u\|_{H^2}^2+\||\bx|^2 u\|_{L^2}^2\right)^{1/2}.
\end{equation}
It will be useful for us that $\Sigma^2$ is a Banach algebra. This space and its generalizations are commonly used in the existence theory of NLS with magnetic potentials, cf., for instance, \cite{Fu, NaSh, Ya}.

Next, we introduce the following nonlinear function,
\begin{equation}\label{eq:F}
\begin{aligned}
F(\theta,u)&:=e^{i\theta \calH}\left(\left|e^{-i\theta \calH}u\right|^{2\sigma}e^{-i\theta \calH}u\right)\\
&=e^{i\theta (\calH-1/2)}\left(\left|e^{-i\theta (\calH-1/2)}u\right|^{2\sigma}e^{-i\theta (\calH-1/2)}u\right),
\end{aligned}
\end{equation}
and study the behavior of $F\left({t}/{\eps^2}, u\right)$, as $\eps \to 0$, where it is readily seen that $F\in C(\RR\times\Sigma^2,\Sigma^2)$. Moreover, in view of \eqref{spec} 
the operator $e^{i\theta (\calH-1/2)}$ is $2\pi$-periodic with respect to $\theta$, hence $F$ is also $2\pi$-periodic with respect to $\theta$. 
Denoting the average of this function by
\begin{equation}\label{eq:average}
\begin{split}
F_{\rm av}(u):= &\, \lim_{T \to \infty} \frac{1}{T} \int_0^T F\left(\theta, u\right) d\theta\\
= &\, \frac{1}{2\pi}\int_0^{2\pi}e^{i\theta \calH}\left(\left|e^{-i\theta \calH}u\right|^{2\sigma}e^{-i\theta \calH}u\right)d\theta,
\end{split}
\end{equation}
the limiting model as $\eps\to 0$ formally reads:
\begin{equation}\label{eqphi}
i\pa_t \phi=-\frac{1}{2} \partial^2_z \phi +  V(z) \phi+ \lambda F_{\rm av}(\phi),
\end{equation}
subject to initial data $\phi(0, \bx)=\psi_0(x,z)$. Here we have used the fact that $\mathcal H$ commutes with the rest of the linear Hamiltonian. 
Equation \eqref{eqphi} describes the resulting averaged particle dynamics. Note that \eqref{eqphi} is still a model in three spatial dimensions. 
In particular, the Gross-Pitaevskii energy associated to \eqref{eqphi} is
\begin{align*}
E(\phi)=&\ \frac{1}{2}\int_{\RR^3}|\partial_z\phi|^2 \,dx \,dz+\int_{\RR^3}V(z) |\phi|^2\, dx\, dz\\
&\ +\frac{\lambda}{2\pi(\sigma+1)}\int_{\RR^3} \int_0^{2\pi}\left|e^{-i\theta \calH}\phi\right|^{2\sigma+2}d\theta \, dx\, dz.
\end{align*}
However, the dependence of the solution to \eqref{eqphi} on $x=(x_1, x_2)$ only stems from the nonlinear averaging operator $F_{\rm av}(\phi)$. 
Thus, in the linear case $\lambda =0$, \eqref{eqphi} becomes a true one-dimensional equation along the {\it unconfined} $z$-axis.

\smallskip

With these notations at hand, we can now state the main result of this work:

\begin{theorem}\label{thmmain}
Let $V$ satisfy \eqref{Vhyp}, $\sigma \in \N$ and $\psi_0 \in \Sigma^2$.

\noindent {\rm (i)} There is a $T_{\rm max}\in (0,\infty]$ and a unique maximal solution $\phi \in C([0,T_{\rm max}),\Sigma^2)\cap C^1([0,T_{\rm max}),L^2(\R^3))$ of the limiting equation \eqref{eqphi}, such that
\[
\| \phi(t, \cdot)\|_{L^2} = \| \psi_0 \|_{L^2}, \quad E(\phi(t, \cdot)) = E(\psi_0), \quad \forall \, t\in [0, T_{\rm max}).
\]
{\rm (ii)} For all $T\in(0,T_{\rm max})$ there are $\eps_T>0$, $C_T>0$ such that, for all $\eps\in(0,\eps_T]$, equation \eqref{NLS} 
admits a unique solution $\psi^\eps\in C([0,T],\Sigma^2)\cap C^1([0,T],L^2(\R^3))$, which is uniformly bounded with respect to $\eps\in (0,\eps_T]$ in $L^\infty((0,T),\Sigma^2)$ and satisfies 
$$\max_{t\in [0,T]}\left\|\psi^\eps(t, \cdot)-e^{-it \calH/\eps^2}\phi(t, \cdot)\right\|_{L^2}\le C_T\,\eps^2.$$
{\rm (iii)} If, in addition, the initial data is concentrated in the $n$-th Landau level $\psi_0 = P_n \psi_0$, then for all $t\in [0, T_{\rm max})$ it holds
$\phi(t) = P_n \phi(t)$ and
\[
i\pa_t \phi=-\frac{1}{2} \partial^2_z \phi +  V(z) \phi+ \lambda P_n \left(|\phi|^{2\sigma} \phi \right).
\]
\end{theorem}

This theorem is in the same spirit as earlier results for NLS with strong anisotropic electric confinement potentials, cf. \cite{delebecque2, bcm, ben2005nonlinear, MeSp}. Similarly, in \cite{delebecque} 
the authors study a Schr\"odinger type model including strong magnetic fields combined with a strong electric confinement. In the present work, however, the 
confinement is {\it solely} due to the {\it magnetic} vector potential $A^\eps$, a situation, which, to the best of our knowledge, has not been studied before in the case of nonlinear Schr\"odinger equations. (For linear 
Schr\"odinger equations, related questions have been considered in the context of the Aharanov-Bohm effect, see, e.g., \cite{IMS}). 
The main qualitative difference between electric and magnetic confinement seems to be that in the latter case, 
the resulting limiting equation \eqref{eqphi} {\it always} remains a model in three spatial dimensions (even after projecting onto the $n$-th Landau level). In particular, it seems 
futile to use an expansion in terms of eigenfunctions of the confining Hamiltonian $\mathcal H$, as is done in earlier works, cf. \cite{delebecque2, ben2005nonlinear}, since in the present situation 
this would result in a system of {\it infinitely many} coupled NLS.

We note that, instead of \eqref{NLSinit}, one might want to consider the analogous equation in only two spatial dimensions, i.e.,
\begin{equation*}
\label{NLSinit2}
i\pa_t \psi=\frac{1}{2}\big( - i \nabla_x  + A^\eps(x) \big)^2 \psi +  \lambda \eps^{2\sigma} |\psi|^{2\sigma}\psi, \quad \psi|_{t=0}=\psi_0(x),
\end{equation*}
where $x= (x_1, x_2)\in \R^2$. If the associated initial data satisfies $\psi_0 = P_n \psi_0(x)$, then the same type of analysis yields the following limiting model:
\[
i\pa_t \phi= \lambda P_n \left(|\phi|^{2\sigma} \phi \right), \quad \psi|_{t=0}=\psi_0(x).
\]
The latter is a generalization of the {\it lowest Landau level equation} (LLL), which is obtained for $\sigma = 1$ and $n=0$. In this case, one usually denotes $\zeta = x_1 +i x_2\in \C$ 
and $P_0$ becomes the orthogonal projector in $L^2(\C)$ on the {\it Bargmann-Fock space} 
\[
\mathcal E = \big\{e^{-|\zeta|^2/2} f(\zeta) \, \text{where $f$ is entire} \big \} \cap L^2(\C).
\]
The LLL equation has been extensively studied, see \cite{aftalion-blanc-nier, Ni}, and, more recently, \cite{GeTh}.

\begin{remark}\label{extension}
It is also possible to generalize our results to include potentials which are of the form $V=V\left(\frac{x}{\eps},z\right)$. In the rescaled variables, the analog of 
\eqref{NLS} then reads
\begin{equation*}\label{NLS2}
i\pa_t \psi^\eps=\frac{1}{\eps^2} \mathcal H \psi^\eps -\frac{1}{2} \partial^2_z \psi^\eps + V(x,z) \psi + \lambda |\psi|^{2\sigma}\psi,\quad \psi|_{t=0}=\psi_0(x,z).
\end{equation*}
By replacing the nonlinear function $F$ with
\begin{equation*}\label{eq:Ftilde}
\widetilde F(\theta,u)=e^{i\theta \calH}\left(\left(V(x,z)+\lambda\left|e^{-i\theta \calH}u\right|^{2\sigma}\right)e^{-i\theta \calH}u\right),
\end{equation*}
our Theorem \ref{thmmain} can be generalized in a straightforward way to obtain the following limit model
\begin{equation*}\label{eqphitilde}
i\pa_t \phi=-\frac{1}{2} \partial^2_z \phi +  \frac{1}{2\pi} \int_0^{2\pi} \widetilde F\left(\theta, \phi\right) d\theta.
\end{equation*}
It is worth noting that 
\[
\frac{1}{2\pi} \int_0^{2\pi} \tilde F(\theta,u) d\theta = \sum_{n\in \N_0} P_n V P_n + \lambda F_{\rm av}(u),
\]
so the first term on the right hand side is diagonal in the Landau levels and thus assertion (iii) of our main theorem remains valid also in this case.
\end{remark}

This paper is now organized as follows. In Section \ref{sec:LWP} below, we shall derive the necessary well-posedness theory for both \eqref{NLS} and \eqref{eqphi}. The 
averaging procedure, which yields assertion (ii) of our main theorem is given Section \ref{sec:aver}. Finally, we discuss the case of the dynamics within 
a given Landau level in Section \ref{sec:land}.

\section{Well-posedness results}\label{sec:LWP}

Before going into the details of the proof of the main theorem, we first need to provide a suitable (local in-time) well-posedness theory for \eqref{NLS} as well as \eqref{eqphi}. To this end, we 
recall the well-posedness results for magnetic NLS proved in \cite{NaSh}, which themselves rely on a construction of the fundamental solution 
of the associated linear Schr\"odinger group by \cite{Ya} and \cite{Fu, Fu1}. Indeed, in view of our assumptions on $A$ and $V$ it is easily seen that \eqref{NLS} 
falls within the class of models studied in \cite{NaSh}, i.e., NLS with smooth, sub-quadratic electric potentials $V$ and magnetic potentials $A\in C^\infty(\R^3;\R^3)$ satisfying:
\[
|\partial^\alpha A(\bx) |\le C_\alpha, \quad \forall |\alpha|\ge 1,
\]
and such that $B_{jk}=\partial_j A_k- \partial_k A_j$ fulfills
\[
|\partial^\alpha B_{jk}(\bx)|\le C_\alpha \langle x \rangle^{-1-\delta}, \quad \forall |\alpha|\ge 1,
\]
for some $\delta >0$. Under this assumptions we have the following local well-posedness result.

\begin{proposition}\label{prop:existNLS}
Let $A$ and $V$ be as above. Then for any $\psi_0 \in \Sigma^2$ there is a maximal existence time $T_1^\eps \in (0,+\infty]$ such that \eqref{NLS} has a unique solution
\[
\psi^\eps\in C([0,T_1^\eps); \Sigma^2) \cap C^1((0,T_1^\eps); L^2(\R^3)),
\]
depending continuously on the initial data $\psi_0$. \end{proposition}

This result is proved, under slightly more general conditions, in \cite[Theorem 1]{NaSh}. (We point out that \cite{NaSh} uses a different, but equivalent norm in 
$\Sigma^2$.) To this end one uses the fact that the linear Schr\"odinger group $S(t) = e^{-it H}$ generated by the magnetic Hamiltonian
\[
H= \frac{1}{2}\big( - i \nabla_\bx  + A(\bx) \big)^2 \psi + V({ z}),
\]
admits space-time Strichartz estimates on some sufficiently small time interval $I\subset \R$, containing the origin. (Global in-time Strichartz estimates cannot be expected, in general, due to the 
possibility of eigenvalues within $\text{spec}\, H$.) 

\begin{remark} The existence time $T_1^\eps>0$ obtained above, in principle could shrink to zero as $\eps\to 0_+$, but it will be a 
consequence of our approximation result that this is, in fact, not the case. 
\end{remark}

The result above does not directly translate to the limiting equation \eqref{eqphi}. The reason for this is that \eqref{eqphi} does not contain the full three-dimensional Laplacian, but only $\partial^2_z$. 
Thus, the dispersive properties of the associated Schr\"odinger group are much weaker in this case, and one cannot expect the 
full range of Strichartz estimates to be available. Nevertheless one can prove the following result, using a classical fixed point argument.

\begin{lemma}\label{lem:existeff}
Let $V$ satisfy \eqref{Vhyp} and $\sigma \in \N$. Then for any $\psi_0 \in \Sigma^2$ there is a maximal existence time $T_{\rm max}\in(0,\infty]$ such that \eqref{eqphi} has a unique solution 
\[
\phi\in C([0,T_{\rm max}); \Sigma^2) \cap C^1((0,T_{\rm max}); L^2(\R^3)),
\]
depending continuously on the initial data $\psi_0$. In addition, the usual conservation laws for the total mass and energy hold, i.e.,
\[
\| \phi(t, \cdot)\|_{L^2} = \| \psi_0 \|_{L^2}, \quad E(\phi(t, \cdot)) = E(\psi_0), \quad \forall \, t\in [0, T_{\rm max}).
\]
\end{lemma}

\begin{proof}
The result follows from a classical fixed-point argument (see, e.g., \cite{caze}) based on Duhamel's formula for the solution $\phi$, i.e.,
\[
\phi(t) = U(t)\psi_0 - i \lambda \int_0^t U(t-s) F_{\rm av} (\phi(s))\, ds=: \Phi(\phi)(t),
\]
where $U(t)=e^{-it \mathcal H_z}$ is the Schr\"odinger group generated by 
\[
\mathcal H_z = -\frac{1}{2}\partial^2_z +V(z),
\]
with smooth, sub-quadratic potential $V(z)$. In order to prove that the map $\Phi$ is a 
contraction on some suitably chosen ball $B_R(0)\subset C([0,T); \Sigma^2)$, we first recall that the results of \cite{Fu, Fu1} show that for any  
$\varphi \in \Sigma^2$: $U(\, \cdot\, )\varphi \in C(\R; \Sigma ^2 )\cap C^1(\R; L^2(\R^3))$. It therefore suffices to show that the nonlinear term $F_{\rm av}$ is
locally Lipschitz to conclude the desired result, cf. \cite[Chapter 6.1]{Pa}. To this end, we first note that for $\sigma \in \N$, the map $z\mapsto |z|^{2\sigma} z$ is smooth and locally Lipschitz. 
This fact directly translates to $F_{\rm av}$ in view of the second line in \eqref{eq:average} and the 
fact that $\Sigma^2$ is a Banach algebra. Continuous dependence on the initial data, as well as the conservation laws for the mass and energy then follow by classical arguments, cf. \cite{caze, Pa}.
\end{proof}

\begin{remark} Unfortunately, the conservation laws of mass and energy are not sufficient to infer global in-time existence of such 
solutions, i.e., $T_{\rm max}=+\infty$, 
even in the defocusing case $\lambda\ge 0$. In order to obtain a global result, one would need to work with solutions in 
$\Sigma^1 = \{u\in H^1(\R^3):\ |\bx| u \in L^2(\R^3)\}$, whose life-span can be controlled by 
the mass and energy. This, however, results in a severe restriction on $\sigma$ when one tries to prove Lipschitz continuity 
of the nonlinearity in three dimensions. 
\end{remark}


\section{A-priori estimates and averaging} \label{sec:aver}

This section is devoted to the proof of Theorem \ref{thmmain} (i) and (ii). Clearly, item (i) is a direct consequence of Lemma \ref{lem:existeff}. 
In the following it will be convenient to work with the norm
\begin{equation*}
\label{eq:defequivnorm}
\| f\|_{\Sigma^2}' := \left( \|f\|^2_{L^2} + \|\mathcal H_0 f\|^2_{L^2} + \| \partial^2_z f \|^2_{L^2} +\| z^2 f \|_{L^2}^2 \right)^{1/2}
\end{equation*}
defined in terms of the harmonic oscillator
$$
\mathcal H_0=-\Delta_x+\frac{1}{4}|x|^2.
$$
One can show (see \cite{He} and also \cite{bcm}) that there are constants $0<C_{_<}<C_{_>}<\infty$ such that for all $f\in\Sigma^2$,
\begin{equation*}
\label{eq:equivnorm}
C_{_<} \|f\|_{\Sigma^2} \le \|f\|_{\Sigma^2}'\le C_{_>} \|f\|_{\Sigma^2} \,.
\end{equation*}

As a preliminary to the proof of item (ii) we investigate the boundedness properties of $e^{-it\calH}$ and $e^{-it\calH_0}$ on the space $\Sigma^2$.

\begin{lemma}\label{bounded}
For all $t\in\R$ and all $f\in\Sigma^2$,
\[
\|e^{-it\calH} f\|_{\Sigma^2}' =\|f\|_{\Sigma^2}'.
\]
\end{lemma}

\begin{proof}
This simply follows from the fact that the operator $\mathcal H$ commutes with all four operators $1$, $\mathcal H_0$, $-\partial_z^2$ and $z^2$ that appear in the definition of the norm $\|\cdot\|_{\Sigma^2}'$. For $\mathcal H_0$ this can be seen by noting that $$\mathcal H=\frac{1}{2}\mathcal H_0-\frac{i}{2}x^\perp\cdot \nabla_x$$ and that $\mathcal H_0$ commutes with $ix^\perp\cdot \nabla_x$.
\end{proof}

\begin{lemma}\label{boundedhz}
There is a $C>0$ such that for all $t\in\R$ and all $f\in\Sigma^2$,
\[
\|e^{-it\calH_z} f\|_{\Sigma^2}'\le e^{C|t|} \|f\|_{\Sigma^2}'.
\]
\end{lemma}

\begin{proof}
Clearly, the term involving $\calH_0$ in the definition of $\|\cdot\|_{\Sigma^2}'$ is invariant under $e^{-it\calH_z}$. We shall prove now that for all sufficiently smooth and rapidly decaying $f$,
\begin{equation}\label{eq:hzdiffineq}
\begin{aligned}
& \frac{d}{dt} \left( \| e^{-it\calH_z} f \|_{L^2}^2 +  \| \partial_z^2 e^{-it\calH_z} f \|_{L^2}^2 + \|z^2 e^{-it\calH_z} f \|_{L^2}^2 \right)  \\
& \qquad \le C \left(  \| e^{-it\calH_z} f \|_{L^2}^2 + \| \partial_z^2 e^{-it\calH_z} f \|_{L^2}^2 +  \| z^2 e^{-it\calH_z} f\|_{L^2}^2 \right).
\end{aligned}
\end{equation}
By Gronwall's lemma this implies that
\begin{align*}
& \| e^{-it\calH_z} f \|_{L^2}^2 + \| \partial_z^2 e^{-it\calH_z} f \|_{L^2}^2 + \|z^2 e^{-it\calH_z} f \|_{L^2}^2 \\
& \quad \le e^{C|t|} \left( \|f\|_{L^2}^2+ \| \partial_z^2 f \|_{L^2}^2 + \|z^2 f \|_{L^2}^2 \right),
\end{align*}
which, by density, extends to all $f\in\Sigma^2$ and, thus, yields the asserted bound on $\|e^{-it\calH_z} f\|_{\Sigma^2}'$.

It remains to prove \eqref{eq:hzdiffineq}. We first compute for sufficiently ``nice'' functions $f$
\begin{align*}
& \frac{d}{dt} \left( \| e^{-it\calH_z} f \|_{L^2}^2 +  \| \partial_z^2 e^{-it\calH_z} f \|_{L^2}^2 + \|z^2 e^{-it\calH_z} f \|_{L^2}^2 \right) \\
& \qquad = i \left( e^{-it\calH_z} f, [\calH_z,\partial_z^4] e^{-it\calH_z} f \right) + i \left( e^{-it\calH_z} f, [\calH_z ,z^4 ] e^{-it\calH_z} f \right) \\
& \qquad = -2\im \left( \partial_z^2 e^{-it\calH_z} f, [\calH_z,\partial_z^2] e^{-it\calH_z} f \right) -2\im \left( z^2 e^{-it\calH_z} f, [\calH_z ,z^2 ] e^{-it\calH_z} f \right).
\end{align*}
Thus, by the Schwarz inequality,
\begin{align*}
& \frac{d}{dt} \left( \| e^{-it\calH_z} f \|_{L^2}^2 +  \| \partial_z^2 e^{-it\calH_z} f \|_{L^2}^2 + \|z^2 e^{-it\calH_z} f \|_{L^2}^2 \right) \\
& \quad \le 2 \| \partial_z^2 e^{-it\calH_z} f\|_{L^2} \left\| [\calH_z,\partial_z^2] e^{-it\calH_z} f \right\|_{L^2} + 2 \| z^2 e^{-it\calH_z} f\|_{L^2} \left\| [\calH_z ,z^2 ] e^{-it\calH_z} f \right\|_{L^2} \\
& \quad \le 2 \left( \| \partial_z^2 e^{-it\calH_z} f\|_{L^2}^2 + \| z^2 e^{-it\calH_z} f\|_{L^2}^2 \right)^{1/2} \\
& \quad \qquad \times \left( \left\| [\calH_z,\partial_z^2] e^{-it\calH_z} f \right\|_{L^2}^2 + \left\| [\calH_z ,z^2 ] e^{-it\calH_z} f \right\|_{L^2}^2 \right)^{1/2}.
\end{align*}
We now compute the commutators
$$
[\calH_z,\partial_z^2] = [V,\partial_z^2] = -2 V'\partial_z - V''
\quad
\text{and}
\quad
[\calH_z ,z^2 ] = -\frac12 [\partial_z^2,z^2 ] = - 2z\partial_z -1.
$$
Thus, clearly,
$$
\left\| [\calH_z ,z^2 ] e^{-it\calH_z} f \right\|_{L^2} \lesssim \left\| e^{-it\calH_z} f \right\|_{\Sigma^2}.
$$
Moreover, in view of our assumption \eqref{Vhyp} on $V$ (with $\alpha=2$),
$$
\left\| [\calH_z,\partial_z^2] e^{-it\calH_z} f \right\|_{L^2}
\lesssim \left\| (|z|+1)\partial_z e^{-it\calH_z} f \right\|_{L^2} + \left\| e^{-it\calH_z} f \right\|_{L^2} \lesssim \left\| e^{-it\calH_z} f \right\|_{\Sigma^2}.
$$
This concludes the proof of \eqref{eq:hzdiffineq} and therefore of the lemma.
\end{proof}

We now begin with the proof of item (ii) in Theorem \ref{thmmain}, following the same strategy as in \cite{MeSp}. We fix $0<T<T_{\rm max}$, where $T_{\rm max}$ is as in item (i) of the theorem, and set
\begin{equation}
\label{defM}
M:=\sup_{\eps>0}\|e^{-it\calH/\eps^2}\phi\|_{L^\infty((0,T)\times \RR^3)}.
\end{equation}
Because of the continuous imbedding $\Sigma^2\hookrightarrow H^2(\R^3)\hookrightarrow L^\infty(\R^3)$, Lemma \ref{bounded} and the existence result in Lemma \ref{lem:existeff} we have
\[
\|e^{-it\calH /\eps^2}\phi\|_{L^\infty((0,T)\times \RR^3)}\le C\|e^{-it\calH/\eps^2}\phi\|_{L^\infty((0,T),\Sigma^2)} = C\|\phi\|_{L^\infty((0,T),\Sigma^2)}<+\infty,
\]
that is, $M<\infty$. In particular, we have $\|\psi_0\|_{L^\infty}=\|\phi(0, \cdot)\|_{L^\infty}\le M.$ Next, we introduce
\begin{equation}
\label{defTeps}
T^\eps:=\sup\left\{t\in [0,T_1^\eps)\,:\, \|\psi^\eps(s)\|_{L^\infty}\le 2M \ \mbox{for all }s\in [0,t]\right\},
\end{equation}
where $T_1^\eps>0$ is local existence time defined in Proposition \ref{prop:existNLS}. We have $T^\eps>0$ by continuity and the fact that $\|\psi^\eps(0)\|\le M$.

\begin{lemma}\label{lem:apriori} 
There is a $C_M>0$ such that for all $\eps>0$ and all $t\in [0, T^\eps]$,
\[\| \psi^\eps(t) \|_{\Sigma^2}' \le \| \psi_0 \|_{\Sigma^2}' \ e^{C_M t}.\]
\end{lemma}

 \begin{proof}
It follows from Lemmas \ref{bounded} and \ref{boundedhz} and the fact that $\calH$ and $\calH_z$ commute that for all $t\in\R$ and $f\in\Sigma^2$,
\begin{equation}
\label{eq:propagator}
\left\| e^{-it(\eps^{-2}\calH+\calH_z)} f \right\|_{\Sigma^2}'\le e^{C|t|} \|f\|_{\Sigma^2}'.
\end{equation}
The crucial point here is that the right side is independent of $\eps$.

Equation \eqref{NLS} for $\psi^\eps$ in Duhamel form reads
\begin{equation*}
\psi^\eps(t) = e^{-it(\eps^{-2}\calH+\calH_z)} \psi_0 - i\lambda \int_0^t e^{-i(t-s)(\eps^{-2}\calH+\calH_z)}|\psi^\eps(s)|^{2\sigma} \psi^\eps(s)\,ds.
\end{equation*} 
Therefore, according to \eqref{eq:propagator}, if $t\geq 0$,
$$
\left\|\psi^\eps(t)\right\|_{\Sigma^2}' \le e^{C t} \left( \|\psi_0\|_{\Sigma^2}' + |\lambda| \int_0^t e^{C s} \left\| |\psi^\eps(s)|^{2\sigma} \psi^\eps(s) \right\|_{\Sigma^2}' ds \right).
$$
Since $\sigma$ is an integer we easily find the following Moser-type inequality,
$$
\left\| |f|^{2\sigma} f \right\|_{\Sigma^2}' \le C' \|f\|_{L^\infty}^{2\sigma} \|f\|_{\Sigma^2}' \,.
$$
Therefore, recalling that $\|\psi^\eps(s)\|_{L^\infty}\le 2M$ if $s\le T^\eps$, we obtain for all $t\in [0,T^\eps]$,
$$
\left\|\psi^\eps(t)\right\|_{\Sigma^2}' \le e^{C t} \left( \|\psi_0\|_{\Sigma^2}' + C' |\lambda| (2M)^{2\sigma} \int_0^t e^{C s} \left\| \psi^\eps(s) \right\|_{\Sigma^2} ds \right).
$$
By Gronwall's lemma this yields
$$
\left\|\psi^\eps(t)\right\|_{\Sigma^2}' \le \|\psi_0\|_{\Sigma^2}' e^{C t+ C' |\lambda | (2M)^{2\sigma} t}
$$
for all $t\in (0,T^\eps]$, which proves the lemma.
\end{proof}

In view of the continuity of $t\mapsto \psi^\eps(t)\in \Sigma^2$, an important consequence of the foregoing lemma is that
\begin{equation}\label{ifthen}
\text{if $T^\eps < \infty$, then $T^\eps < T_1^\eps$ and $\|\psi^\eps(T^\eps)\|_{L^\infty} = 2M$,}
\end{equation}
a fact we shall use in the proof below.

\begin{proof}[Proof of Theorem \ref{thmmain}(ii)]
We consider the filtered unknown $\phi^\eps := e^{i t \mathcal H/\eps^2} \psi^\eps$, which satisfies
\begin{equation*}
i\pa_t \phi^\eps= -\frac{1}{2} \partial^2_z \phi^\eps + V(z) \phi + \lambda F\left(\frac{t}{\eps^2}, \phi^\eps\right),\quad \phi^\eps|_{t=0} = \psi_0,
\end{equation*}
where $F$ is defined in \eqref{eq:F}. Denoting the difference $$u^\eps := \phi^\eps - \phi$$ for $0\le t < \min( T_{\rm max}, T_1^\eps)$ and recalling that $\phi$ solves \eqref{eqphi} we obtain that $u^\eps$ satisfies
\begin{align*}
u^\eps(t)&=\lambda \int_0^t U(t-s)\left(F\left(\frac{s}{\eps^2},\phi^\eps(s)\right)-F\left(\frac{s}{\eps^2},\phi(s)\right)\right)ds\\
&\quad +\lambda \int_0^t U(t-s)\left(F\left(\frac{s}{\eps^2},\phi(s)\right)-F_{\rm av}(\phi(s))\right)ds\\
&\equiv A_1+A_2.
\end{align*}
Here we denote, as before, the Schr\"odinger group $U(t)=e^{-i t \mathcal H_z}$, and we have also used the fact that $u^\eps(0) = 0$. 

In order to estimate $A_1$, we recall that for $0\le t\le \min(T,T^\eps)$, \eqref{defM} and \eqref{defTeps} imply that
\begin{align*}
\|{ A_1}\|_{L^2}&\lesssim\int_0^t\left\||\psi^\eps(s)|^{2\sigma}\psi^\eps(s)-|e^{-is{ \calH}/\eps^2}\phi(s)|^{2\sigma}e^{-is{ \calH}/\eps^2}\phi(s)\right\|_{L^2}ds\\
&\lesssim\int_0^t\left(\|\psi^\eps(s)\|_{L^\infty}^{2\sigma}+\left\|e^{-is{ \calH}/\eps^2}\phi(s)\right\|_{L^\infty}^{2\sigma}\right)\|\psi^\eps(s)-e^{-is{ \calH}/\eps^2}\phi(s)\|_{L^2}ds\\
&\lesssim M^{2\sigma}\int_0^t\|u^\eps(s)\|_{L^2}ds \\
&\lesssim \int_0^t\|u^\eps(s)\|_{L^2}ds.
\end{align*}
Here and in the rest of this proof we use the convention that the implied constant in $\lesssim$ may depend on $T$, but not on $\eps$.

On the other hand, in order to estimate $A_2$ we introduce the following function, defined on $\RR\times \Sigma^2$,
$$\mathcal F(\theta,u)=\int_0^\theta (F(s,u)-F_{\rm av}(u))ds,$$
and write as in \cite{bcm, MeSp} 
\begin{align*}
& U(t-s)\left(F\left(\frac{s}{\eps^2},\phi(s)\right)-F_{\rm av}(\phi(s))\right)\\
&=\eps^2\frac{d}{ds}\left(U(t-s)\mathcal F\left(\frac{s}{\eps^2},\phi(s)\right)\right)
-i\eps^2 U(t-s)\mathcal H_z\mathcal F\left(\frac{s}{\eps^2},\phi(s)\right) \\
& \quad -\eps^2 U(t-s)D_u\mathcal F\left(\frac{s}{\eps^2},\phi(s)\right)\cdot\pa_t\phi(s),
\end{align*}
where we recall that $\mathcal H_z= -\frac{1}{2}\partial^2_z +V(z)$. Then we can bound
\begin{align*}
\|A_2\|_{L^2}&\le \eps^2|\lambda|\left\|\mathcal F\left(\frac{t}{\eps^2},\phi(t)\right)\right\|_{L^2}+\eps^2|\lambda|\int_0^t\left\| \mathcal H_z\mathcal F\left(\frac{s}{\eps^2},\phi(s)\right)\right\|_{L^2}ds\\
&\quad +\eps^2|\lambda|\int_0^t\left\|D_u\mathcal F\left(\frac{s}{\eps^2},\phi(s)\right)\cdot\pa_t\phi(s)\right\|_{L^2}ds.
\end{align*}
In order to bound the right-hand side, 
we note that, since $F(\cdot,u)$ is $2\pi$-periodic and $F_{\rm av}$ is its average, $\mathcal F(\cdot,u)$ is also $2\pi$-periodic. Hence, it is readily seen that this function satisfies the following properties,
\begin{align*}\mbox{if}\quad \|u\|_{\Sigma^2}\le R,\quad &\mbox{then}\quad \sup_{\theta\in \RR}\|\mathcal F(\theta,u)\|_{\Sigma^2}\le CR^{2\sigma+1},\\
\mbox{if}\quad \|u\|_{\Sigma^2}+\|v\|_{L^2}\le R,\quad &\mbox{then}\quad \sup_{\theta\in \RR}\|D_u\mathcal F(\theta,u)\cdot v\|_{L^2}\le CR^{2\sigma+1}.
\end{align*}
Since $\phi\in L^\infty([0,T],\Sigma^2)$, $\pa_t\phi\in L^\infty([0,T],L^2(\R^3))$ and since $\|\calH_z u\|_{L^2}\lesssim \|u\|_{\Sigma^2}$ (by our assumptions \eqref{Vhyp} on $V$), we can finally bound
$$
\|A_2\|_{L^2}\lesssim \eps^2.
$$
(Here, we have also used the fact that the time interval has at most length $T\lesssim 1$.)

In summary, we have proved that, for all $t\in[0,\min(T,T^\eps)]$,
$$\|u^\eps(t, \cdot)\|_{L^2}\lesssim \eps^2+\int_0^t\|u^\eps(s, \cdot)\|_{L^2}ds.$$
Thus, Gronwall's lemma yields that, for all $t\in[0,\min(T,T^\eps)]$,
\begin{equation}
\label{estieps}
\|\psi^\eps(t)-e^{-it\calH/\eps^2}\phi(t)\|_{L^2} = \|\phi^\eps(t)-\phi(t)\|_{L^2} = \| u^\eps (t) \|_{L^2} \lesssim \eps^2.
\end{equation}
We consequently deduce from \eqref{defM}, a Gagliardo--Nirenberg inequality, \eqref{estieps}, Lemma \ref{bounded} and Lemma \ref{lem:apriori} that, for all $t\in[0,\min(T,T^\eps)]$,
\begin{align*}
\|\psi^\eps(t)\|_{L^\infty}&\le M+\|\psi^\eps(t)-e^{-it\calH/\eps^2}\phi(t)\|_{L^\infty}\\
&\le M+C\|\psi^\eps(t)-e^{-it\calH/\eps^2}\phi(t)\|_{L^2}^{1/4} \ \|\psi^\eps(t)-e^{-it\calH/\eps^2}\phi(t)\|_{H^2}^{3/4}\\
&\le M+C'\eps^{1/2}\left(\|\psi^\eps(t)\|_{\Sigma^2}'+\|\phi(t)\|_{\Sigma^2}'\right)^{3/4}\\
&\le M+C''\eps^{1/2}.
\end{align*}
Here the constant $C''$ depends on $T$, but not on $\eps$. (Note that the factor $e^{C_M t}$ in Lemma \ref{lem:apriori} can be bounded by $e^{C_M T}$ for $t\le\min(T,T^\eps)$.)

Hence, for $\eps< \eps_T:=(M/C'')^2$, we have 
\begin{equation}
\label{bornLinf}
\forall t\le \min(T,T^\eps),\quad \|{ \psi}^\eps(t)\|_{L^\infty}< 2M.
\end{equation}
We claim that this implies $T^\eps\ge T$. In fact, this is trivial when $T^\eps=\infty$ and otherwise we deduce from \eqref{ifthen} that $\|\psi^\eps(T^\eps)\|=2M$, which contradicts \eqref{bornLinf}.

Consequently, \eqref{estieps} is valid on $[0,T]$ which proves the inequality in item (ii) of Theorem \ref{thmmain}. Finally the claimed uniform boundedness in $L^\infty((0,T),\Sigma^2)$ with respect to $\eps\in(0,\eps_T]$ follows from Lemma \ref{lem:apriori} and the fact that $T^\eps\ge T$. This completes the proof.
\end{proof}


\section{Dynamics within a single Landau level} \label{sec:land}

In this section, we first prove that for initial data $\psi_0$ concentrated within a given Landau level, the effective dynamics is given by item (iii) of Theorem \ref{thmmain}. 
To this end, we denote by $P_n=P_n^2$ the spectral projection onto the $n$-th eigenspace of $\mathcal H$, and assume that initially $\psi_0 = P_n \psi_0$. 

\begin{proof}[Proof of Theorem \ref{thmmain}(iii)]
We let $P_n^\bot=1-P_n$ and $w:=P_n^\bot \phi \in C([0,T_{\rm max}),\Sigma^2)$, and recall that $w(0, \bx) = 0$, since $\phi(0,\bx) = P_n\psi_0(\bx)$. 
It suffices to show that $w(t, \bx)=0$, for all $t\in [0, T_{\rm max})$. 
To this end, we first note that the equation satisfied by $w$ is
\begin{align*}
i\pa_t w&=\mathcal H_z w+   \frac{\lambda }{2\pi}\int_0^{2\pi}P_n^\bot F(\theta,\phi)d\theta,
\end{align*}
where $F$ is defined by \eqref{eq:F}. We rewrite this equation as
\begin{align*}
i\pa_t w&=\mathcal H_z  w +  \frac{ \lambda}{2\pi}\int_0^{2\pi}P_n^\bot F(\theta,P_n\phi) \, d\theta \\
&\quad +   \frac{ \lambda}{2\pi}\int_0^{2\pi}P_n^\bot \left(F(\theta,P_n\phi+w)-F(\theta,P_n\phi)\right) \, d\theta,
\end{align*}
and also note that
\begin{align*}
F(\theta,u)&=e^{i\theta \calH}\left(\left|e^{-i\theta \calH}u\right|^{2\sigma}e^{-i\theta \calH}u\right)\\
&=e^{i\theta (\calH-n-1/2)}\left(\left|e^{-i\theta (\calH-n-1/2)}u\right|^{2\sigma}e^{-i\theta (\calH-n-1/2)}u\right).
\end{align*}
Using this, we see that
\begin{align*}\int_0^{2\pi} P_n^\bot F(\theta,P_n\phi)d\theta&=\int_0^{2\pi}P_n^\bot e^{i\theta (\calH-n-1/2)}\left(\left|P_n\phi\right|^{2\sigma}P_n\phi\right)\,d\theta\\
&=\sum_{m\neq n}\left(\int_0^{2\pi}e^{i\theta (m-n)}d\theta\right) P_m\left(\left|P_n\phi\right|^{2\sigma}P_n\phi\right)\\
&=0,
\end{align*}
due to the $2\pi$-periodicity of $e^{i\theta (m-n)}$. 
By writing the Duhamel formulation of the equation $w$, we therefore obtain
$$w(t)=-i\lambda \int_0^te^{-i(t-s)\mathcal H_z}\frac{1}{2\pi}\int_0^{2\pi}P_n^\bot\left(F(\theta,P_n\phi(s)+w(s))-F(\theta,P_n\phi(s))\right) d\theta ds,$$
and therefore
\begin{align*}
& \|w(t)\|_{L^2}\le \, \frac{|\lambda|}{2\pi} \int_0^t\int_0^{2\pi}\left\| F(\theta,P_n\phi(s)+w(s))-F(\theta,P_n\phi(s)) \right\|_{L^2}\, d\theta \, ds \\
& \qquad \le C\int_0^t \sup_\theta \left( \| e^{-i\theta\calH} (P_n\phi(s)+w(s)) \|_{L^\infty}^{2\sigma} + \| e^{-i\theta\calH} P_n\phi(s) \|_{L^\infty}^{2\sigma} \right) \|w(s)\|_{L^2} \, ds.
\end{align*}
Since $\Sigma^2 \hookrightarrow   H^2(\R^3) \hookrightarrow L^\infty (\R^3)$, we can use Lemma \ref{bounded} to obtain
$$
\| e^{-i\theta\calH} (P_n\phi(s)+w(s)) \|_{L^\infty} = \| e^{-i\theta\calH}\phi(s) \|_{L^\infty} \lesssim \| e^{-i\theta\calH} \phi(s) \|_{\Sigma^2}' =  \|\phi(s)\|_{\Sigma^2}'.
$$
According to item (i) in Theorem \ref{thmmain} this is bounded on any interval $[0,T]$ with $T<T_{\rm max}$. Similarly, we bound
$$
\| e^{-i\theta\calH} P_n\phi(s) \|_{L^\infty} \lesssim \| e^{-i\theta\calH} P_n\phi (s)\|_{\Sigma^2}' \le \| P_n\phi(s) \|_{\Sigma^2}'.
$$
We now obtain the same bound as before if we use that
$$
\| P_n\phi(s) \|_{\Sigma^2}' \le \| \phi(s) \|_{\Sigma^2}'.
$$
The proof of this inequality is similar to the proof of Lemma \ref{bounded}. In fact, the inequality is obvious for all terms in the definition of the norm $\|\cdot\|_{\Sigma^2}'$ except for the one involving $\calH_0$. As observed in Lemma \ref{bounded}, however, $\calH_0$ commutes with $\calH$ and therefore also with $P_n$. Thus, $\|\calH_0 P_n u\|_{L^2}= \|P_n \calH_0 u\|_{L^2}\le \|\calH_0 u\|_{L^2}$, as claimed.

To summarize, we have shown that for any $T<T_{\rm max}$ there is a $C$ such that for all $t\in [0,T]$,
$$
\|w(t)\|_{L^2}\le C\int_0^t\|w(s)\|_{L^2} \, ds.
$$
By Gronwall's lemma, we deduce that $w\equiv 0$ on $[0,T]$ for any $T<T_{\rm max}$. This completes the proof of the theorem.
\end{proof}

\bibliographystyle{amsplain}

\end{document}